\documentclass[a4paper,10pt]{article}
\usepackage{amsmath,amsfonts,amssymb,latexsym}
\usepackage{amsthm}
\usepackage{dsfont}	
\usepackage{booktabs}
\usepackage{ctable}
\usepackage{multirow}
\usepackage{morefloats}
\usepackage{float}
\restylefloat{figure}
\usepackage{caption}
\usepackage{subcaption}
\graphicspath{ {Figures_New/} } 
\usepackage[font=small,labelfont=bf]{caption}
\usepackage{rotating}
\usepackage{marvosym}
\usepackage[english]{babel}
\usepackage{pstricks,graphicx}
\usepackage{xcolor}

\theoremstyle{plain}

\newtheorem{proposition}{Proposition}[section]

\theoremstyle{definition}

\theoremstyle{remark}
\newtheorem{remark}{Remark}[section]

\DeclareMathOperator{\1}{\mathds{1}}

\title{Pricing Reliability Options under different electricity prices' regimes}

\author{Luisa Andreis\footnote{WIAS-Weierstrass Institute, Berlin; andreis@wias-berlin.de }, Maria Flora\footnote{Department of Economics, University of Verona; maria.flora@univr.it}, Fulvio Fontini\footnote{Department of Economics and Management "Marco Fanno", University of Padua, and  Interdepartmental Centre for Energy Economics and Technology "Giorgio Levi-Cases", University of Padua; fulvio.fontini@unipd.it}, Tiziano Vargiolu\footnote{Department of Mathematics "Tullio Levi Civita", University of Padua and  Interdepartmental Centre for Energy Economics and Technology ``Giorgio Levi-Cases", University of Padua; vargiolu@math.unipd.it}}

\begin{document}

\maketitle

\begin{abstract}
Reliability Options are capacity remuneration mechanisms aimed at enhancing security of supply in electricity systems. They can be framed as call options on electricity sold by power producers to System Operators. This paper provides a comprehensive mathematical treatment of Reliability Options. Their value is first derived by means of closed-form pricing formulae, which are obtained under several assumptions about the dynamics of electricity prices and  strike prices. Then, the value of the Reliability Option is simulated under a real-market calibration, using data of the Italian power market. We finally perform sensitivity analyses to highlight the impact of the level and volatility of both power and strike price, of the mean reversion speeds and of the correlation coefficient on the Reliability Options' value.
\end{abstract}


\newpage

\section{Introduction}
Capacity Remuneration Mechanisms (CRM) have been implemented in several electricity markets worldwide, in order to remunerate explicitly power capacity.\footnote{See \cite{CF}, ch. 22 and 23, for an introduction and an analysis of Capacity Remuneration Mechanisms.} Among these, the Reliability Option (RO) mechanism recognizes the option nature of the investments in power capacity  
and creates a market for such an option. ROs, firstly proposed in \cite{Bid,VRP}, have been implemented in Colombia (Firm Energy Obligations \cite{Cramton and Stoft}), in NE-ISO (Forward Capacity Market \cite{FERC}) and in Ireland \cite{SEMc,SEMb,SEMa} and are about to be implemented in Italy \cite{MFRB,Ternaa,Ternab,Ternad}. They are tools to commercialize, through a financial product, the possibility, given by generation capacity,  
of providing security of supply by producing electricity. They give their holder, i.e. the System Operator (SO), which acquires them in a competitive setting, the right to call the generation capacity to produce power, and to receive the positive difference between the electricity price that effectively occurs in the market and a pre-defined price. Such a pre-defined price corresponds to the strike price of the option, and it is set to represent the value that power has at the specific level for which load is not shed, i.e., it is the highest system marginal price compatible with load provision with no load shedding. 

In this paper, we evaluate ROs following the financial approach\footnote{This is a standard approach to price financial derivatives, see for instance \cite{Hull}.}, which requires to identify the stochastic property of the asset under evaluation, and to assume that, in a complete market, a continuous hedging between the financial derivative and the underlying asset is possible. At a first glance, this assumption seems quite hard to be met in the electricity sector, given that the underlying asset of the option is electricity, which is not a storable good\footnote{At least as long as storage of electric energy by means of conversion into a different form of energy, such as kinetic energy of water in power dams or as chemical energy in batteries, is limited because of its cost or for technical reasons.}. However, derivatives can be and are indeed written on several underlying assets that are not liquidly traded, such as interest rates or temperatures (see e.g.~\cite[Chapter 15]{Bjork}). What is needed for the application of risk-neutral pricing based on hedging is the existence of liquid assets that are traded and that correlate with the underlying of the derivative, such as forwards. The seminal paper \cite{BesLem} has questioned such an assumption, considering the relationship between derivative (future) prices and spot prices in markets with limited liquidity and risk averse agents. However, we believe that the assumption of limited liquidity was more justified at the beginning of the liberalization process of the power market, while this concern is less justified now, after several years of functioning of liberalized electricity markets. This approach is shared by other scholars, who have evaluated exotic options on electricity, such as spark-spread options (options on the differential between power prices and the heat content of the fuel, \cite{Deng,Hik}), Asian options (options written on average prices, \cite{Clewlow}) or options which are implicit in demand response mechanisms, \cite{Sezgen}.

We formulate different possible assumptions on the dynamics of the stochastic processes on which the RO depends, and estimate the relative RO value. ROs are complex options on power supply which can have different maturities and can be exercised several times at different, and possibly random, strike prices. Therefore, we provide a comprehensive mathematical treatment of all their aspects. 

Though many authors, as seen above, have evaluated various exotic options on electricity,
 to the best of our knowledge our paper is the first one to evaluate ROs under different assumptions on the  electricity price process.\footnote{ \cite{Burger} evaluates, through a Monte Carlo approach, a contract composed of a portfolio of 4344 call options on hourly prices, all with the same strike price. It corresponds to a discrete-time version of the option we consider in Proposition \ref{Prop:GBM}.} 
Several models for electricity prices have been proposed in the literature (see e.g.~\cite{BenKoe,Clewlow,GemRon,Hik,Paraschiv,Veh} and the book \cite{Benth} for a presentation and critical discussion of various models), and it would not be feasible to present RO pricing formulae for each one of these. For this reason, we choose a set of simple and significant ones, and present semi-explicit pricing formulae that have clear economic interpretations.
We first start from the simplest possible assumption about electricity prices and strike prices, increasing then the level of complexity of the RO design, to allow for a mean reverting underlying, for stochastic strike prices and for possibly negative (but bounded from below) electricity prices. Furthermore, we simulate the RO value under different possible assumptions on the parameters, and calibrate the RO parameters against real electricity market data, namely, the Italian Power Exchange ones. The availability of long hourly price time series and the forthcoming introduction of RO in the Italian market both justify the choice.  

The paper is structured as follows. Section \ref{sec:rel_op} describes ROs and presents a general pricing formula under realistic assumptions. Section \ref{sec:pricing} provides semi-explicit solutions to the general pricing formula, for different electricity and strike price models.
We start by defining the arbitrage-free boundaries of RO's evaluation. We then move from the  simplistic model of geometric Brownian motion (GBM) with deterministic strike, to correlated GBMs with stochastic strike, and, by increasing realism on the model, to the case when both electricity and strike prices are seasonal and mean-reverting. For all these models, we present semi-explicit pricing formulae. Finally, we provide some insights for the case of negative prices. In Section \ref{sec:simulation}, we showcase a simulation of the RO evaluation and perform a sensitivity analysis, using data of the Italian Power market for estimates and calibration. Section \ref{sec:conclusion} draws conclusions, while all the proofs of the propositions are in the Appendix.

\section{Reliability options} \label{sec:rel_op}

We start by desribing in general what ROs are. These contracts are sold in an auction, typically once a year, and they aim to deliver electricity with a  given $T_1$-length period in advance (lead time), for a pre-defined  period of delivery, which has length $(T_2-T_1)$. The rules of the RO specify that the capacity provider, the subject who sells the option, must commit to deliver a certain capacity to the subject buying the option, in general the SO. Such a commitment is made effective by prescribing that the seller must offer in the market an amount of electricity equal to the committed capacity and return any positive difference between the reference market price and a previously set strike price $K$. Each RO contract scheme specifies what the reference market is. In a first approximation, the reference market can be a convex combination of different markets, such as the day-ahead and the balancing or real-time ones. In practice, different RO schemes can have different reference markets. For instance, in Ireland, exclusively the day-ahead market is taken as a reference, while in NE-ISO it is the real-time one.
If we call $P$ the day-ahead market price and $P^{(b)}$ the price in the balancing market (or in the real-time market), we can define the reference market price $R$ as the following convex combination
$$ R = \lambda P + (1 - \lambda) P^{(b)}, $$
where, as said,  $\lambda \in [0,1]$  depends on the country: $\lambda = 0$ for ISO New England; $\lambda = 1$ for Colombia and Ireland; $\lambda \in [0,1]$ in the case of Italy (see~\cite{MFRB} for a description of the forthcoming Italian market). 

The strike price is in general determined by taking into account the variable costs of the reference peak technology, that is, the dispatchable technology that would be included in the optimal generation mix with the lowest unitary investment cost. In actual RO markets, the rule for the strike price is communicated to potential sellers of ROs before the auction takes place. Thus, in some implementations it can be treated as a deterministic and constant parameter. However, it is also possible that the strike price changes over time during the life span of the RO. This is a possibility envisaged, for instance, in the forthcoming Italian RO scheme, where it is established that the rule linking the strike price to a reference marginal technology is set before the auction, but the marginal cost of such a technology is computed every given period (a month) during the life span of the RO.\footnote{See \cite{MFRB} and \cite{Ternaa, Ternab, Ternad}.} This implies that the strike price can also be conceived as a stochastic process. 
We shall first derive the RO value starting with the simplest case, and then increase the level of complexity, to derive a general representation of the value of the RO.

\subsection{A simple mathematical model for Reliability Options}

The mathematical modeling of the general RO is quite complex, as many auctions and prices are involved. We simplify it by defining a mathematical model for the case when the reference price is simply the day-ahead price $P$, i.e. $\lambda=1$, as it is in the Colombian or the Irish CRM.\footnote{Moreover, we do not consider congestion in the transmission network, and therefore we implicitly assume that the market for ROs have the same size of the electricity market, namely, that there are no differences between the pricing zones of the electricity and the capacity markets.}
In this way, only one state variable is needed for the reference market price $R$, and it is indeed $P$. 

We start by computing the fair price of a RO, written only on the reference price $P$ and based on a generation capacity, i.e., for a power plant that is already in place. As said at the beginning of this section, the RO is sold in an auction at a certain time, but it becomes active in a subsequent time period. Let us denote by $t = 0$ the auction time and by $[T_1,T_2]$, with $T_1 > 0$, the time period when capacity has to be committed. It is assumed that the power plant will be productive at least until $T_2$. The idea of pricing the RO is to compute the expected operational profits at time $t = 0$ (auction time) of the power plant over the period $[T_1,T_2]$, both in the case when the capacity provider enters a RO scheme, and in the case it does not. The difference between these two operational profits will be the fair price of the RO. 

We work on a filtered probability space $(\Omega, \mathcal{F},\left\{\mathcal{F}_t\right\}_{t\geq0},\mathbb{Q})$ such that the probability measure $\mathbb{Q}$ is the risk-neutral pricing measure used by the market, and the day-ahead electricity price $P = (P_t)_{t\geq0}$ is a $\mathbb{Q}$-semimartingale.  We consider the simple case of a thermal plant, with total capacity $Q>0$,\footnote{$Q$ is to be interpreted as the available capacity of a power plant, as described by {\cite{JoskowTirole}}. In the real-world examples of ROs, available capacity is computed by measuring the average availability of a power plant over a given time span (usually a year) and derating the nominal capacity accordingly (as suggested in the academic literature by \cite{CramtonOckenfelsStoft}, and in practical market implementations in \cite{Ternaa, Ternab} for the Italian scheme, and in \cite{SEMderating} for Ireland). As an example, consider a 100MW plant with a maintenance period of one month per year. Its capacity factor is equal to 0.91; this figure can be used to de-rate the relevant capacity of the plant for the RO, which would amount to 91MW.} that converts a fuel, for example oil, gas or coal, into electricity.
The cost $C = (C_t)_{t\geq0}$ of running the thermal plant summarizes the fuel price, CO$_2$ price, operational and other costs.
The power plant sells the electricity at time $t\geq0$ when it wins the day-ahead auction, i.e. when its bid $b_t$ is less than or equal to $P_t$. 
We adopt the usual simplifications, continuous time instead of hourly granularity and no ramping penalties/constraints. The plant can decide its bid process $b = (b_t)_{t\geq0}$ to maximize its revenues. 

We first evaluate the expected operational profits of the power plant over $[T_1,T_2]$ in the case when a RO scheme is not in place, 
This is the \emph{value of the power plant} $V(T_1,T_2)$ at $t=0$ and it depends on the power plant's income over $[T_1,T_2]$. It can be defined as
\begin{equation} \label{linear}
V(T_1,T_2) = \sup_{b\in\mathcal{B}} \mathbf{E}^\mathbb{Q}\left[ \left. \int_{T_1}^{T_2} e^{- r t } Q \mathds{1}_{b_t \leq P_t} (P_t - C_t) dt \right| \mathcal{F}_{0} \right]\,,
\end{equation}
where $\mathcal{B}$ is the set of adapted processes on $[T_1,T_2]$, $r$ is the instantaneous risk-free rate of return  and $\mathbf{E}^\mathbb{Q}$ is the expectation with respect to $\mathbb{Q}$. 

\begin{remark}
In this setting, we assume that the investor is risk-neutral. Although here we are not evaluating financial assets, but rather incomes coming from industrial activity, this is in line with all the related literature (see e.g. \cite{Deng,McDonaldSiegel,Sezgen}), and is justified by the following financial argument. The underlying assets $P$ and $C$ could be in principle not storable, or even not traded in some markets. However, even in such a situation, the risk-neutral evaluation in Eq.~\eqref{linear} can be applied as long as one can find hedging instruments that can be storable and liquidly traded, and that are correlated with $P$ and $C$: for the mathematical derivation of such a result, see e.g.~\cite[Chapter 15]{Bjork} for vanilla products like call and put options (as we will end up to have), and \cite[Remark 3.6]{CCGV} for structured products like that in Eq.~\eqref{linear} and the subsequent ones\footnote{This is exactly the same argument used to evaluate derivative assets written on non-tradable quantities like interest rates, temperature, etc.}. Here, we indeed have such suitable hedging instruments, i.e. forward contracts on power and fuel (for $P$ and $C$, respectively), which are liquidly traded on financial markets, as they are basically equivalent to any other financial asset up to few days before physical delivery. When physical delivery approaches, in order to maintain the hedging position it is sufficient to liquidate the position on the maturing future(s) and open an equivalent new one on another future with a physical delivery  farther in time. This is a standard practice in energy markets, called ``rolled-over portfolios", see e.g.~\cite{Alexander,ETV} for two applications.
\end{remark}

Going back to Eq.~\eqref{linear}, it is optimal to choose $b$ such that $\mathds{1}_{b_t \leq P_t} = 1$ if and only if $P_t > C_t$, i.e. the optimal bidding process is $b_t = C_t \, \forall \, t\in[T_1,T_2]$. Thus, the final payoff for a thermal plant is
$$ V(T_1,T_2) =\mathbf{E}^\mathbb{Q}\left[\left.  Q  \int_{T_1}^{T_2} e^{- r t}  (P_t - C_t)^+ dt\right| \mathcal{F}_{0} \right]\,. $$
We now consider the case when the thermal plant writes a RO with strike price $K=(K_t)_{t\geq0}$. The plant must now pay back $(P_t - K_t)^+$. Therefore, the value $ V_{ro}(T_1,T_2)$ of the thermal plant with a RO scheme in place is
$$ V_{ro}(T_1,T_2) = \sup_{b\in\mathcal{B}} \mathbf{E}^\mathbb{Q}\left[ \left. \int_{T_1}^{T_2} e^{- r t} Q ( \1_{b_t \leq P_t} (P_t - C_t) - (P_t - K_t)^+)\ dt \right| \mathcal{F}_{0} \right]\,. $$ 
The bidding strategy $b_t = C_t$ is again optimal  for all $ t\in[T_1,T_2]$. Thus,
$$ V_{ro}(T_1,T_2) = V(T_1,T_2) - \mathbf{E}^\mathbb{Q}\left[ \left. \int_{T_1}^{T_2} e^{- r t} Q (P_t - K_t)^+\ dt \right| \mathcal{F}_{0} \right] \,.$$ 
In a risk-neutral world, the value $RO(T_1,T_2)$ of a RO written on the time interval $[T_1,T_2]$ should make the investor indifferent between having the original plant without the RO, and having it with the RO written on it plus the price of the option, i.e.
$ V(T_1,T_2) = V_{ro}(T_1,T_2) + RO(T_1,T_2)\,.
$ 
Therefore, the final result is 
\begin{eqnarray}
RO(T_1,T_2) & = & V(T_1,T_2) - V_{ro}(T_1,T_2) \nonumber \\
& = & \mathbf{E}^\mathbb{Q}\left[ \left. \int_{T_1}^{T_2} e^{- r t} Q (P_t - K_t)^+\ dt \right| \mathcal{F}_{0} \right] \label{form_of_price}
\end{eqnarray} 
Thus, the value of a reliability option issued by a thermal plant is equivalent to the price of an insurance contract against price peaks. Interestingly enough, notice that the operating strategy of the power plants does not change. In electricity markets, it is well known that perfectly competitive markets without CRMs, the so called energy only markets, provide enough incentives to investment, and the same is true for optimally designed CRMs, since the latter simply anticipate ex ante the supermarginal profits that investors would gain in energy only markets. In other words, the amount of remuneration of capacity accruing from perfectly competitive markets for CRMs equals the expected discounted value of the supermarginal profits gained in electricity markets; in a world without market failures, the two levels coincide (see \cite[Chapter 22]{CF}). This is confirmed in our framework: without market power, the value of operating the plant is independent of the form of remuneration of power production, i.e., if revenues accrue ex-ante from the CRM or ex-post from selling electricity in the market.

\section{Pricing of Reliability Options} \label{sec:pricing}

\subsection{Model-free no-arbitrage bounds}

Equation \eqref{form_of_price} already allows us to produce model-free no-arbitrage bounds on the price of the RO. No-arbitrage bounds have been derived by \cite{Deng} for analogous contracts, yet in a different setting. In fact, in \cite{Deng}, it is assumed that a continuum of forward contracts is traded, both for electricity and for the relevant fuel (whose spot price here is $K$), which deliver at any given date $t$. However, many energy markets do not satisfy this assumption, and especially forward contracts on electricity, which guarantee the delivery of power over a period (e.g., $[T_1,T_2]$), rather than on a single date $t$. 
Even if this does not jeopardize the evaluation mechanism developed in Section \ref{sec:rel_op}, as these forwards written on a period are liquid assets that can be used for hedging, the no-arbitrage bounds available in \cite{Deng} cannot be directly applied to our framework, but must be modified. Notice that these model-free bounds do not require any assumption on the electricity price apart from $P$ being bounded from below by a constant price floor $-P^*$, with $P^*\geq 0$. This is consistent with those electricity markets in which negative prices are allowed with a lower bound (as for instance in the German and French markets).

We start from the identity
$$ (P_t - K_t)^+ = (K_t - P_t)^+ + P_t - K_t\,. $$
 Since $0 \leq (K_t - P_t)^+ \leq K_t+P^*$, we have  
$$ P_t - K_t \leq (P_t - K_t)^+ \leq P_t+P^* \,.$$
By multiplying the inequalities by $e^{-rt}$, integrating and taking the expectation, we have that 
$$ Q \mathbf{E}^\mathbb{Q}\left[ \left. \int_{T_1}^{T_2} e^{- r t} (P_t - K_t)\ dt \right| \mathcal{F}_{0} \right] \leq RO(T_1,T_2) \leq Q \mathbf{E}^\mathbb{Q}\left[ \left. \int_{T_1}^{T_2} e^{- r t} (P_t+P^*)\ dt \right| \mathcal{F}_{0} \right] \,.$$
The right-hand side represents the forward price of delivering the quantity $Q$ of electricity over the period $[T_1,T_2]$\footnote{this is alternatively referred to as {\em flow forward} or {\em swap}, see e.g.\cite{Benth}.} with an additional constant $QP^* \frac{e^{-rT_1} - e^{-rT_2}}{r}$, depending on the price floor. We label
$$ F_P(0;T_1,T_2) := \mathbf{E}^\mathbb{Q}\left[ \left. \int_{T_1}^{T_2} e^{- r t}P_t\ dt \right| \mathcal{F}_{0} \right]  $$
the (unitary) forward price. 
Then, since $RO(T_1,T_2) \geq 0$, when $K_t \equiv K$, i.e. with fixed strike, we can rewrite the no-arbitrage relation above as
\begin{equation} \label{NA-fixed}
Q \left( F_P(0;T_1,T_2) - K \frac{e^{-rT_1} - e^{-rT_2}}{r} \right)^+ \leq RO(T_1,T_2) \leq Q F_P(0;T_1,T_2) +QP^* \frac{e^{-rT_1} - e^{-rT_2}}{r}\,.
\end{equation}
Thus, the value of a reliability option written on a total capacity $Q$ over the period $[T_1,T_2]$ lies between the intrinsic value of $Q$ call options on the forward $F_P(0;T_1,T_2)$ and the modified strike $K \frac{e^{-rT_1} - e^{-rT_2}}{r}$, and $Q$ forwards $F_P(0;T_1,T_2)$ adjusted by an additional constant proportional to the price floor $P^*$.

Conversely, when $K$ follows itself a stochastic process, we define
$$ F_K(0;T_1,T_2) := \mathbf{E}^\mathbb{Q}\left[ \left. \int_{T_1}^{T_2} e^{- r t} K_t\ dt \right| \mathcal{F}_{0} \right] \,,$$
and obtain
\begin{equation} \label{NA-variable}
Q \left( F_P(0;T_1,T_2) - F_K(0;T_1,T_2) \right)^+ \leq RO(T_1,T_2) \leq Q F_P(0;T_1,T_2)+QP^* \frac{e^{-rT_1} - e^{-rT_2}}{r} \,.
\end{equation}
Note that, even with a stochastic strike price $K$, the upper bound is unaffected. On the other hand, the lower bound is now the intrinsic value of $Q$ exchange options on the forward $F_P(0;T_1,T_2)$ for the forward $F_K(0;T_1,T_2)$.

The advantage of these no-arbitrage bounds lies in the fact that, even though there is not a forward contract traded on the market for the total period $[T_1,T_2]$, this period is usually a multiple of calendar years, whose contracts are commonly traded. For example, in the Italian RO design, the period $[T_1,T_2]$ starts on January, 1 of year $Y$ and lasts until December, 31 of year $Y + 2$: thus, $F_P(0;T_1,T_2)$ ends up  simply being the sum of the three calendar products for the years $Y$, $Y + 1$ and $Y + 2$. In the case when the stochastic strike $K$ is indexed with some marginal technology determined  in advance (e.g.~combined cycle gas turbines), analogous forward contracts possibly exist for the corresponding fuel (gas in this case). 

The no-arbitrage bounds above are model-free, in the sense that they hold for any no-arbitrage model that we specify in the following for the dynamics of $P$, and possibly of $K$, the only assumption needed being the existence of a price floor for $P$. However, to evaluate the RO as a financial contract, it is necessary to specify the stochastic process modeling electricity prices. The electricity price shows peculiarities that make it difficult to model, such as strong seasonality and mean-reversion. For this reason, several processes have been adopted to reproduce the price dynamics. In what follows, we provide semi-explicit formulae to price a RO over $[T_1,T_2]$ under different price dynamics. Note that the price models generally used to evaluate options do not allow for negative prices. We will use models of this kind in the subsequent sections, while allowing for negative prices in Section 3.6 below.

\subsection{Electricity spot price as a geometric Brownian motion}\label{Sec:GBM}

Let us start with the simplest assumption, i.e. that  the price of electricity $P$ evolves as a GBM, and that the option's strike price $K$ is a fixed deterministic value. We stress that the former is an assumption that we already know is unreasonable, in the sense that it cannot be assumed to provide a realistic representation of the electricity price dynamics. However, it is the simplest possible assumption that is used to derive explicit pricing formulae for call options. Thus, we treat it as a first simplified approach to help us presenting the main features of the model. In this case, the price $P$, under the risk-neutral measure $\mathbb{Q}$, is assumed to be the solution of the following SDE:
\begin{align}
dP_t=&r P_t dt+\sigma P_t dB_t,\label{GBM_price}
\end{align}
where $B$ is a one-dimensional $\mathbb{Q}$-Brownian motion and $r$ is the instantaneous risk-free rate of return.

The price of a RO in this case is equivalent to the time integral over the interval $[T_1,T_2]$ of a European call option with strike price $K$ and maturity ranging in $[T_1,T_2]$. 
In the following proposition, we provide a semi-explicit formula to price the RO, under the assumptions above. 
\begin{proposition}\label{Prop:GBM} 
Let the reference market price $P$ follow the dynamics~\eqref{GBM_price}. The price of a reliability option over the time interval $[T_1,T_2]$ with fixed strike price $K\geq0$ is given by the following formula:
\begin{align}
RO(T_1,T_2)=&\int_{T_1}^{T_2} Q\left[P_0 N(d_1(K,P_{0},t))-e^{-r t}K N(d_2(K,P_{0},t))\right]dt \,, \label{GBM_price_formula_BS}
\end{align}
where $N$ is the cumulative distribution function (CDF) of a standard Gaussian random variable and
\begin{align*}
d_1(K,P_{0},t) \colon = &\frac{1}{\sigma\sqrt{t}}\left[\ln\left(\frac{P_{0}}{K}\right)+\left(r+\frac{\sigma^2 }{2}\right)t\right]\,,\\
d_2(K,P_{0},t)\colon  = & d_1(K,P_{0},t)-\sigma \sqrt{t}\,.
\end{align*}

\end{proposition}
Proposition \ref{Prop:GBM} simply uses the Black and Scholes formula, since $RO(T_1,T_2)$ can be defined as the time integral of a family of call options with the same underlying and strike price, indexed by their maturity in $[T_1,T_2]$.\footnote{Interestingly enough, this result solves also a problem firstly posed in \cite{McDonaldSiegel}, in the framework of firms' evaluations.} Thus, it provides a formula that can be applied to compute the value of the RO, once the parameters upon which the call depends on have been set; namely, the risk-free interest rate $r$, the starting price $P_{0}$ and the electricity price volatility $\sigma$.

\subsection{Electricity price and strike price as correlated Geometric Brownian Motions}\label{Sec:GBM_corr}

A first step to increase the level of complexity consists in modeling the strike price as a stochastic process. Recall that, in ROs, the strike price is the marginal cost of the marginal technology. Complex RO schemes can allow it to change over time, according to a predefined rule. For instance, it can be assumed that the strike price is given by the fuel cost of a predefined marginal technology, such as Combined Cycle Gas Turbines. In such a way, the strike price will be linked to a reference fuel price. Alternatively, it can be established that the reference price changes at fixed regular dates according to a given indexing formula, for example monthly, and stays constant in each of these sub periods.\footnote{As mentioned, this is going to be the case of the future Italian RO scheme.} 
Both cases imply that the strike price is a stochastic process. Thus, a first extension of the model defined in Section~\ref{Sec:GBM} is to model $K$ and $P$ as two (possibly correlated) geometric Brownian motions. This means that the prices $(K_t,P_t)_{t\geq0}$ follow a risk-neutral dynamics of the following type:
\begin{equation} \label{corr_GBM}
\left\{\begin{array}{l}
dK_t = (r-q_k) K_t dt+\sigma_k K_t dB^1_t \,,\\
dP_t = (r-q_p) P_t dt+\sigma_p P_t dB^2_t\,,
\end{array}\right.
\end{equation}
where $(B^1,B^2)$ are correlated $\mathbb{Q}$-Brownian motions, with correlation $\rho\in [-1,1]$. 
Notice that the correlation of the two stochastic processes depends on the rules defining the strike price and on the strike price nature. For instance, if the variable strike price is set to be equal to the marginal cost of the marginal technology, and if the electricity market is perfectly competitive, the system marginal price will be equal to the marginal cost of the marginal technology. Thus, the correlation coefficient would be equal to $1$. If, on the contrary, the stochastic strike price equals some weighted average of different marginal costs at different hours, for instance at peak and off-peak hours, then the correlation coefficient would be positive but less than $1$, since the electricity price $P$ would be more volatile than the strike price $K$. Finally, it is also possible that the strike price is negatively correlated with the electricity price, depending on how the strike price is defined and on what reference basket it is linked to. However, this possibility is rather unlikely, for the reasons mentioned above. 

The following proposition provides the value of the RO with two GBMs:
\begin{proposition}\label{Prop:corr_GBM} 
Let the reference market price $P$ and the RO strike price $K$ follow the dynamics~\eqref{corr_GBM}. Then the price of a reliability option over the time interval $[T_1,T_2]$ is given by
\begin{equation}\label{price_corr_GBM}
RO(T_1,T_2)=
\int_{T_1}^{T_2} \left(P_{0}e^{-q_p t}N(a_1(K_{0},P_{0},t)) - K_{0} e^{-q_k t} N(a_2(K_{0},P_{0},t)) \right) dt\,,
\end{equation}
where $N$ is the CDF of a standard normal random variable,  and
\begin{align*}
a_1(K_{0},P_{0},t) \colon = & \frac{\ln\left(\frac{P_{0}}{K_{0}}\right)+(q_p-q_k)t}{\sigma \sqrt{t}}  + \frac{1}{2} \sigma \sqrt{t}\,,\\
a_2(K_{0},P_{0},t) \colon = & a_1(K_{0},P_{0},t)-\sigma \sqrt{t}\,, \\
\sigma \colon = & \sqrt{\sigma_k^2+\sigma_p^2-2\rho\sigma_k\sigma_p} = \sqrt{(\sigma_k - \sigma_p)^2 + 2 (1 -\rho) \sigma_k \sigma_p}\,.
\end{align*}
\end{proposition}

In contrast to Proposition \ref{Prop:GBM}, in Proposition \ref{Prop:corr_GBM} we used the Margrabe formula with dividends (see, for instance, \cite{CarDur}), instead of the Black-Scholes one.
Here, the $RO(T_1,T_2)$ value is equal to the time integral of a family of options to exchange the (random) electricity price $P$ with the (random) strike price $K$, again indexed by their maturity. As usual in the Margrabe formula, the relevant volatility is $\sigma$, that can be interpreted as the volatility of the ratio $P/K$ (i.e., of the electricity price expressed in units of the strike price), which is decreasing with respect to the correlation $\rho$. In particular, for $\rho \to 1$ (i.e.~when the strike price is highly correlated with the electricity price), we have $\sigma \to |\sigma_k - \sigma_p|$. In this case, when also $\sigma_k = \sigma_p$, the volatility vanishes, and the value of the option is determined just by its intrinsic value. Instead, for $\rho \to - 1$ (i.e. when the strike price is highly negatively correlated with the electricity price), we have $\sigma \to \sigma_k + \sigma_p$, i.e., the volatility is maximized. However, we stress that this latter case is rather unlikely for the case of RO, as typically a stochastic strike price $K$ is defined in terms of quantities related to electricity generation (as e.g.~the marginal price of the marginal technology, or some related market index), so that we should expect a positive correlation.

\subsection{Mean-reverting electricity price with seasonality}\label{Sec:season_1}

As mentioned, a GBM does not capture typical stylized facts of electricity prices, namely seasonality and mean-reversion. A natural extension is thus to price the RO when the dynamics of the reference price reflects the aforementioned stylized facts. In particular, we  model the log-spot price of electricity as a mean-reverting process encoding different types of seasonality by means of a time-dependent function. This approach has been widely adopted in energy markets, see for instance \cite{Benth} and references therein. We first assume a deterministic strike price. In the next section, we shall remove this assumption.

We define the function describing seasonality trends for all $t\geq0$, as
\begin{equation} \label{eq:seasonality}
\mu(t)= \alpha + \sum_{i=1}^{12} \beta_i \,month_i(t) + \sum_{i=1}^{4} \delta_i \,day_i(t) + \sum_{i=1}^{24} \gamma_i\, hour_i(t) \,,
\end{equation}  
where $month_i(t)$, $day_i(t)$ and $hour_i(t)$ are dummies for month, day of week and hour, used to capture different types of seasonality. Specifically, we assume that $day$ can take 4 values: `Friday', `Weekend', `Monday', and `other working day'. This captures the differences between working days and weekend as well as possible first- or end-of-the-working-week effect.

We then consider the day-ahead price $P$ as 
\begin{equation}\label{Price_seasonality}
P_t = 
 e^{\mu(t)}e^{X_t} \, ,
\end{equation}
where $X_t$, under the risk-neutral measure $\mathbb{Q}$, is the solution of the SDE
\begin{align}
dX_t=&-\lambda X_t dt+\sigma dW_t \, ,\label{OU_log_price}
\end{align}
where $W$ is a one dimensional $\mathbb{Q}$-Brownian motion, $\sigma$ stands for the volatility and $\lambda>0$ is the mean-reversion speed.

We have the following:
\begin{proposition}\label{Prop:OU} 
Let the reference market price $P$ follow the dynamics~\eqref{eq:seasonality}--\eqref{Price_seasonality}--\eqref{OU_log_price}. Then the price of a reliability option over the time interval $[T_1,T_2]$ with fixed strike price $K\geq0$ is given by 
\begin{align}\label{OU_price}
RO(T_1,T_2) = & Q \int_{T_1}^{T_2} e^{-r t} \left[  f(0,t) N(d_1(K,P_{0},t)) - K N(d_2(K,P_{0},t))\right] dt \, ,
\end{align}
where $N$ is the CDF of a normal random variable, $P_0 = e^{\mu(0) + X_0}$  
and
\begin{align*}
f(0,t) \colon = & {\mathbf E}[P_t | {\mathcal F}_0] = \exp\left( \mu(t) + e^{-\lambda t} + \frac12 Var(t) \right), \\
Var(t) \colon = & \frac{\sigma^2}{2\lambda}(1-e^{-2\lambda t}), \\
d_{1,2}(K,P_{0},t) \colon =& \frac{1}{\sqrt{Var(t)}} \log \frac{f(t,T)}{K} \pm \frac12 \sqrt{Var(t)},
\end{align*}
where, by abuse of notation we mean that the definition of $d_{1}(K,P_{0},t)$ involves the $+$ sign and the definition of $d_{2}(K,P_{0},t)$ involves the $-$ sign.
\end{proposition}

\begin{remark}
Equation \eqref{OU_price} is a generalization of Equation \eqref{GBM_price_formula_BS}: in fact, if we let $\mu(t) := (r - q_p - \frac12 \sigma^2) t$ and $\lambda \to 0$, then we reobtain at the limit the model of the previous section. In fact, we have that $m_t \equiv X_0$, $Var(t) \to \sigma^2 t$, 
$$ e^{-r t} f(0,t) \to e^{(r - q_p)t + X_0}\,, $$
and
$$ d_1(K,P_0,t) \to \frac{1}{\sigma \sqrt{t}} \left(X_0 + (r - q_p) t - \frac12 \sigma^2 t - \ln K \right) = \frac{1}{\sigma \sqrt{t}} \ln \frac{e^{X_0 + (r - q_p) t}}{K} - \frac12 \sigma \sqrt{t} \,.$$
Thus, the pricing formula in Equation \eqref{OU_price} collapses into that of Equation \eqref{GBM_price_formula_BS}. 
\end{remark} 

\subsection{Allowing for mean-reverting strike price with seasonality}\label{Sec:season_2}

As a natural extension of the model in Section~\ref{Sec:season_1}, we now consider the case when the strike $K$ is a mean-reverting process (with seasonality) as well. 
The dynamics of the state variables then becomes
\begin{equation}\label{prices_elec_12}
\left\{\begin{array}{ll}
P_t = & e^{\mu(t)} e^{X_t} \,,\\
K_t = & e^{\nu(t)} e^{Y_t}\,.
\end{array}\right.
\end{equation}
Here, $\mu$ is given by~\eqref{eq:seasonality} and $\nu$ is a seasonality function for $K$ of the same form, while the processes $X$ and $Y$ are solution to
\begin{equation}\label{2OU_log_price}
\left\{
\begin{array}{ll}
dX_t=&-\lambda_x X_t dt+\sigma_x dW^1_t\,,\\
dY_t=&-\lambda_y Y_t dt+\sigma_y dW^2_t\,,
\end{array}
\right.
\end{equation}
where $(W^1,W^2)$ are correlated $\mathbb{Q}$-Brownian motions, with correlation $\rho\in [-1,1]$.

\begin{proposition}\label{Prop:2OU} 
Let the reference market price $P$ and the RO strike price $K$ follow the dynamics~\eqref{prices_elec_12}; then the price of a reliability option over the time interval $[T_1,T_2]$ is given by 

\begin{align}\label{2OU_price}
RO(T_1,T_2) = & Q \int_{T_1}^{T_2} e^{-r t} \left( f_P(0,t) N\left(d_2(K_0,P_0,t)\right) -  f_K(0,t) N\left(d_1(K_0,P_0,t)\right)\right) dt \, ,\end{align}
where $N$ is the CDF of a normal random variable, $P_0 = e^{\mu(0) + X_0}$, $K_0 = e^{\nu(0) + Y_0}$  
and

\begin{eqnarray}
f_P(0,t) & \colon = & {\mathbf E}[P_t | {\mathcal F}_0] = \exp\left( \mu(t) + e^{-\lambda_x t} + \frac{\sigma_x^2}{2\lambda_x}(1-e^{-2\lambda_x t}) \right), \\
f_K(0,t) & \colon = & {\mathbf E}[K_t | {\mathcal F}_0] = \exp\left( \nu(t) + e^{-\lambda_y t)} + \frac{\sigma_y^2}{2\lambda_y}(1-e^{-2\lambda_y t}) \right), \\
d_{1,2}(K_0,P_{0},t) & \colon = & \frac{1}{\sqrt{\overline{Var}(t)}} \log \frac{f_P(0,t)}{f_K(0,t)} \pm \frac12 \sqrt{\overline{Var}(t)}, \\
\overline{Var}(t) &\colon = & \sigma_x^2 \frac{1 - e^{-2\lambda_x t}}{2 \lambda_x} + \sigma_y^2 \frac{1 - e^{-2\lambda_y t}}{2 \lambda_y} - 2 \rho \sigma_x \sigma_y \frac{1 - e^{-(\lambda_x+\lambda_y) t}}{\lambda_x + \lambda_y}\, . \label{barvar}
\end{eqnarray}

\end{proposition}

This result is a similar to that of Proposition \ref{Prop:OU} in the same sense as 
Proposition \ref{Prop:corr_GBM} is similar to Proposition \ref{Prop:GBM}: here $RO(T_1,T_2)$ can be again defined as the time integral of a family of options to exchange the electricity price $P$ with the strike price $K$. Here too, the relevant volatility is $\overline{Var}(t)$, which can again be interpreted as the volatility of the ratio $P/K$ (i.e., the electricity price expressed in units of the strike price: this is made explicit in the proof in the Appendix), which is again decreasing with respect to the correlation $\rho$. In particular, for $\rho \to 1$ (i.e.~when the strike price is highly correlated with the electricity price), and $\lambda_x = \lambda_y =: \lambda$ (i.e.~when the two mean-reversion speeds are the same), we have $\overline{Var}(t) \to \frac{1 - e^{-2\lambda t}}{2 \lambda} (\sigma_x - \sigma_y)^2 $. In this case, when $\sigma_x = \sigma_y$, the volatility vanishes, and the value of the option is given just by its intrinsic value. Instead, in the unlikely case (see the discussion at the end of Section \ref{Sec:GBM_corr}) when $\rho \to - 1$  and $\lambda_x = \lambda_y =: \lambda$, we have $\overline{Var}(t) \to  \frac{1 - e^{-2\lambda t}}{2 \lambda} (\sigma_x + \sigma_y)^2$, i.e., the volatility is maximized. 

\subsection{Possible extension to negative day-ahead and strike prices}\label{Sec:season_2_neg}

In principle, it is possible to allow for negative power prices, since we know this is a possibility in energy markets, see~\cite{EGV} and references therein. An analogous extension can be also envisaged for strike prices, especially when these are linked to power prices.
A possible approach to model negative prices is to set negative values $-P^*$ and $-K^*$, for certain $P^*,K^* \geq 0$, as price floors for $P$ and $K$, respectively, and to consider the following shifted dynamics 
\begin{equation}\label{prices_elec_neg}
\left\{\begin{array}{ll}
P_t = & \left(e^{\mu(t)} e^{X_t} - P^*\right)\,,\\
K_t = & \left(e^{\nu(t)} e^{Y_t} - K^*\right)\,.
\end{array}\right.
\end{equation}
where $\mu$ and $\nu$ are again seasonality functions for $P$ and $K$ and the processes $X$ and $Y$ are solution of Equation \eqref{2OU_log_price}, in analogy with the previous section. 

By setting $C := P^*-K^*$, one can prove that the price of the reliability option is now given by the following expression:
\begin{equation}
RO(T_1,T_2) = Q   \int_{T_1}^{T_2}e^{- r t}  \mathbf{E}^{\mathbb{Q}}\left[\left. (e^{\mu(t)}e^{X_t} - e^{\nu(t)}e^{Y_t} - C)^+\  \right| \mathcal{F}_{0} \right]dt\,.
\end{equation}
The above formula is the time integral of a family of spread options with a fixed strike price $C$ and indexed by their expiration date in $[T_1,T_2]$. Therefore, considering dynamics of type \eqref{prices_elec_neg} relates the problem of pricing a Reliability Option to the problem of pricing a spread option (see \cite{CarDur} for a survey of classical frameworks and methods for spread options).  
Unfortunately, a general closed formula for the pricing of spread options is not available. 
However, since the RO is in principle a quite illiquid product, one can use a numerical method to price it in this general case, for example Monte Carlo.

\section{Simulation and sensitivity analysis}\label{sec:simulation}

In this section we simulate the value of the RO under realistic assumptions on the parameter values. To do so, we fit the parameters of the electricity price dynamics to a real market, using data of the Italian market. For simplicity, we consider day-ahead prices only, and use the weighted average of Italian zonal prices, called PUN (\textit{Prezzo Unico Nazionale}), ranging from January 1 to December 31, 2016.

As previously explained, we used dummies to capture monthly, daily and hourly seasonality, as defined in Eq. \eqref{eq:seasonality}. We chose `January', `Friday' and `hour 1' as reference groups, against which the comparisons are made.
Figure \ref{fig:price_and_seasonality} shows the calibrated seasonality function, plotted against the historical PUN data.
Furthermore, we considered an annual risk-free rate $r = 0.01$ and, in the pricing models where the only stochastic variable is the electricity price, we considered $K=40$  \EUR/MWh. According to the scheme to be implemented in Italy, the pricing of the RO starts 4 years from now, and the option has a maturity of 3 years (\(T_1=4, \, T_2 = 7\)). 

The starting point $X_0$ is taken equal to 0.
Table \ref{tab:estimates} reports the estimated parameters for each different model, while Table \ref{tab:season_estimates} 
shows the estimated seasonality parameters.

\begin{figure}[H]
\caption{Seasonality function in \eqref{eq:seasonality} (solid red line, upper panel) calibrated on historical 2016 PUN electricity data (solid blue line, upper panel) and residuals (bottom panel).\label{fig:price_and_seasonality}}

\centering
	\begin{subfigure}[b]{0.75\textwidth}
	\includegraphics[width=\textwidth]{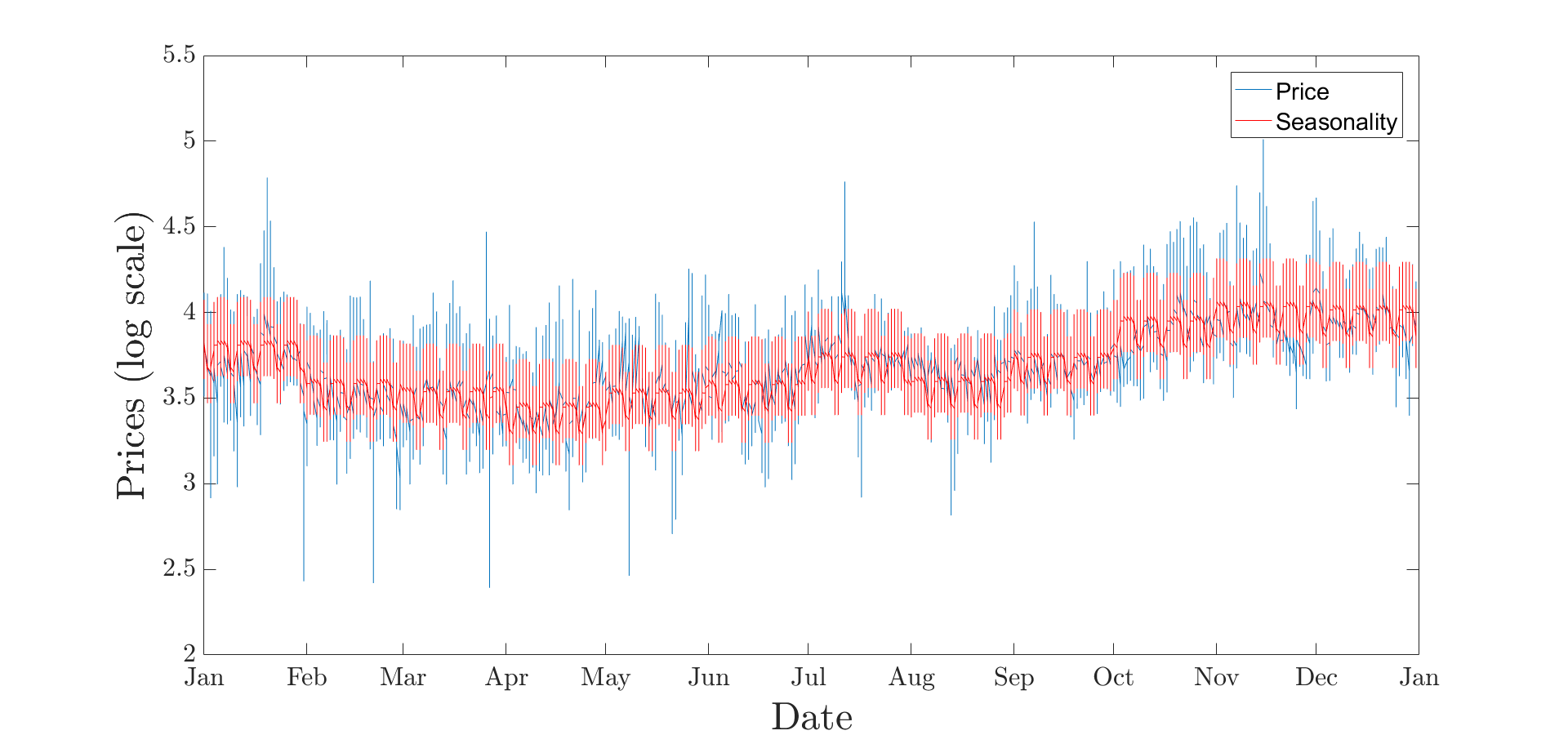}
	\end{subfigure}
\quad
	\begin{subfigure}[b]{0.75\textwidth}
	\includegraphics[width=\textwidth]{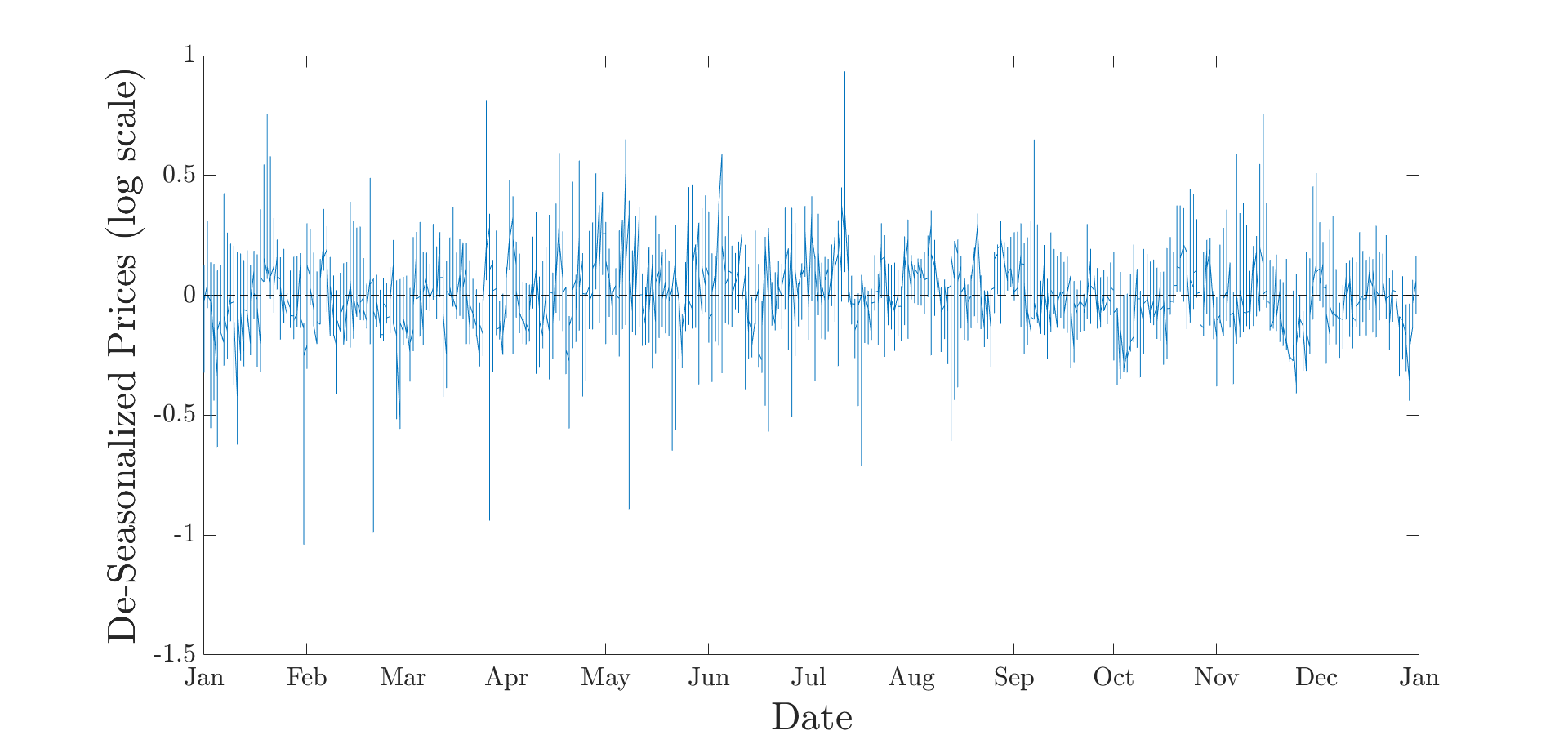}
	\end{subfigure}
\end{figure}

\begin{table}[h]
\centering
\small
\begin{tabular}{cccc}  
\toprule
 & GBM & 1-OU  & 2-OU  \\
\midrule
 $\hat{\sigma}$ & 5.4041 & 6.5932 & 6.5932\\
\cmidrule(r){1-1}
  $\hat{\lambda}$ & - & 294.84 & 294.84\\
 \bottomrule
\end{tabular}
\caption{Estimated yearly parameters $\hat{\sigma}$ and $\hat{\lambda}$ for each pricing model (electricity price following a Geometric Brownian motion (GBM), electricity price following a mean-reverting Ornstein-Uhlenbeck process (1-OU), correlated electricity and strike prices following mean-reverting Ornstein-Uhlenbeck processes (2-OU)).\label{tab:estimates}}

\end{table}

\begin{table}[H]
\centering
\small
\begin{tabular}{ccccccccc}  
\toprule
 & Estimate & S.E. & pValue &  &  & Estimate & S.E. & pValue \\
\cmidrule(r){1-4}	\cmidrule(r){6-9}
 Intercept & \( 3.79 \) & \( 0.01 \) & \( 0 \) &  &  $ hour_{6} $ & \( -0.13 \) & \( 0.01 \) & \( 0 \) \\
\cmidrule(r){1-1}	\cmidrule(r){6-6}
  $ month_2 $ & \( -0.22 \) & \( 0.01 \) & \( 0 \) &  &  $ hour_{7} $ & \( -0.01 \) & \( 0.01 \) & \( 0.5 \) \\
  \cmidrule(r){1-1}	\cmidrule(r){6-6}
  $ month_3 $ & \( -0.27 \) & \( 0.01 \) & \( 0 \)  & &  $ hour_{8} $ & \( 0.1 \) & \( 0.01 \) & \( 0 \) \\
  \cmidrule(r){1-1}		\cmidrule(r){6-6}
  $ month_4 $ & \( -0.36 \) & \( 0.01 \) & \( 0 \) & &  $ hour_{9} $  & \( 0.18 \) & \( 0.01 \) & \( 0 \) \\
  \cmidrule(r){1-1}		\cmidrule(r){6-6}
  $ month_5 $ & \( -0.28 \) & \( 0.01 \) & \( 0 \)  & &  $ hour_{10} $  & \( 0.16 \) & \( 0.01 \) & \( 0 \) \\
  \cmidrule(r){1-1}		\cmidrule(r){6-6}
  $ month_6 $ & \( -0.23 \) & \( 0.01 \) & \( 0 \)   & &  $ hour_{11} $  & \( 0.12 \) & \( 0.01 \) & \( 0 \) \\
  \cmidrule(r){1-1}		\cmidrule(r){6-6}
  $ month_7 $ & \( -0.07 \) & \( 0.01 \) & \( 0 \)   & &  $ hour_{12} $    & \( 0.07 \) & \( 0.01 \) & \( 0 \)  \\
  \cmidrule(r){1-1}		\cmidrule(r){6-6}
  $ month_8 $ & \( -0.21 \) & \( 0.01 \) & \( 0 \)     & &  $ hour_{13} $  & \( 0 \) & \( 0.01 \) & \( 0.8 \)  \\
  \cmidrule(r){1-1}		\cmidrule(r){6-6}
  $ month_9 $ & \( -0.07 \) & \( 0.01 \) & \( 0 \)      & &  $ hour_{14} $  & \( -0.05 \) & \( 0.01 \) & \( 0 \) \\
  \cmidrule(r){1-1}		\cmidrule(r){6-6}
  $ month_{10} $ & \( 0.14 \) & \( 0.01 \) & \( 0 \)     & &  $ hour_{15} $  & \( -0.02 \) & \( 0.01 \) & \( 0.13 \) \\
  \cmidrule(r){1-1}		\cmidrule(r){6-6}
  $ month_{11} $ & \( 0.23 \) & \( 0.01 \) & \( 0 \)    & &  $ hour_{16} $  & \( 0.04 \) & \( 0.01 \) & \( 0 \) \\
  \cmidrule(r){1-1}		\cmidrule(r){6-6}
  $ month_{12} $ & \( 0.21 \) & \( 0.01 \) & \( 0 \)     & &  $ hour_{17} $  & \( 0.09 \) & \( 0.01 \) & \( 0 \) \\
  \cmidrule(r){1-1}		\cmidrule(r){6-6}
  Monday &   \( -0.01 \) & \( 0.01 \) & \( 0.04 \)     & &  $ hour_{18} $  & \( 0.15 \) & \( 0.01 \) & \( 0 \) \\
    \cmidrule(r){1-1}		\cmidrule(r){6-6}
  Weekend &   \( -0.14 \) & \( 0.01 \) & \( 0 \)     & &  $ hour_{19} $  & \( 0.22 \) & \( 0.01 \) & \( 0 \) \\
    \cmidrule(r){1-1}		\cmidrule(r){6-6}
  Working\_day  & \( 0.02 \) & \( 0.01 \) & \( 0 \)      & &  $ hour_{20} $  & \( 0.28 \) & \( 0.01 \) & \( 0 \) \\
  \cmidrule(r){1-1}		\cmidrule(r){6-6}
  $ hour_{2} $  & \( -0.08 \) & \( 0.01 \) & \( 0 \)   & &  $ hour_{21} $ & \( 0.27 \) & \( 0.01 \) & \( 0 \)  \\
  \cmidrule(r){1-1}		\cmidrule(r){6-6}
  $ hour_{3} $  & \( -0.15 \) & \( 0.01 \) & \( 0 \)   & &  $ hour_{22} $ & \( 0.2 \) & \( 0.01 \) & \( 0 \)  \\
  \cmidrule(r){1-1}		\cmidrule(r){6-6}
  $ hour_{4} $ & \( -0.18 \) & \( 0.01 \) & \( 0 \) & & $ hour_{23} $ & \( 0.12 \) & \( 0.01 \) & \( 0 \)  \\
  \cmidrule(r){1-1}		\cmidrule(r){6-6}	
  $ hour_{5} $ & \( -0.18 \) & \( 0.01 \) & \( 0 \) & &   $ hour_{24} $ & \( 0.03 \) & \( 0.01 \) & \( 0.01 \)  \\
   \bottomrule
\end{tabular}
\caption{Linear regression estimates, standard errors and p-values obtained using the specification in \eqref{eq:seasonality}. The base group categories for each dummy variable are $month_1$, $friday$ and $hour_1$.}
\label{tab:season_estimates}
\end{table}

As mentioned, real electricity prices do not follow GBMs. Therefore, in the simulation, we start from the model defined in Section \ref{Sec:season_1}.

\subsection{Mean reverting electricity price with seasonality, fixed strike}

We simulate the value of the RO using the Monte Carlo methodology. Specifically, we compute the RO value using 10,000 simulations of the price path of the underlying.

\begin{figure}[h!]
	\centering
	\begin{subfigure}[b]{0.45\textwidth}
	\includegraphics[width=\textwidth]{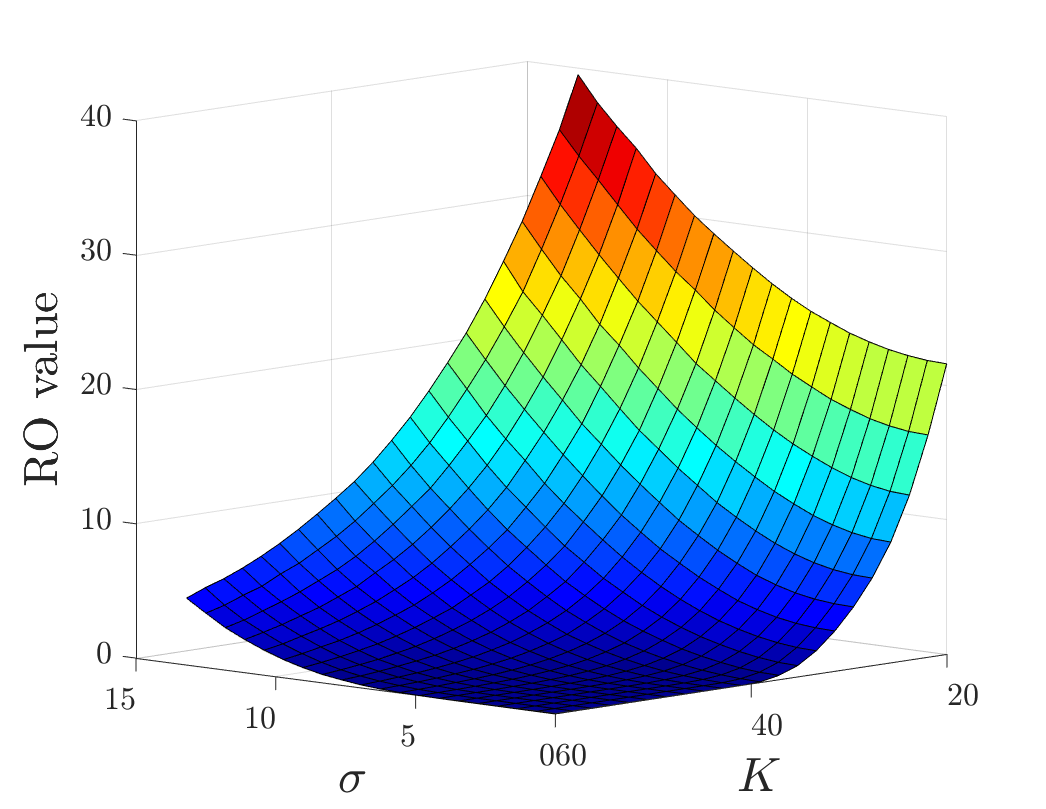}
		\end{subfigure}
	\quad
	\begin{subfigure}[b]{0.45\textwidth}
	\includegraphics[width=\textwidth]{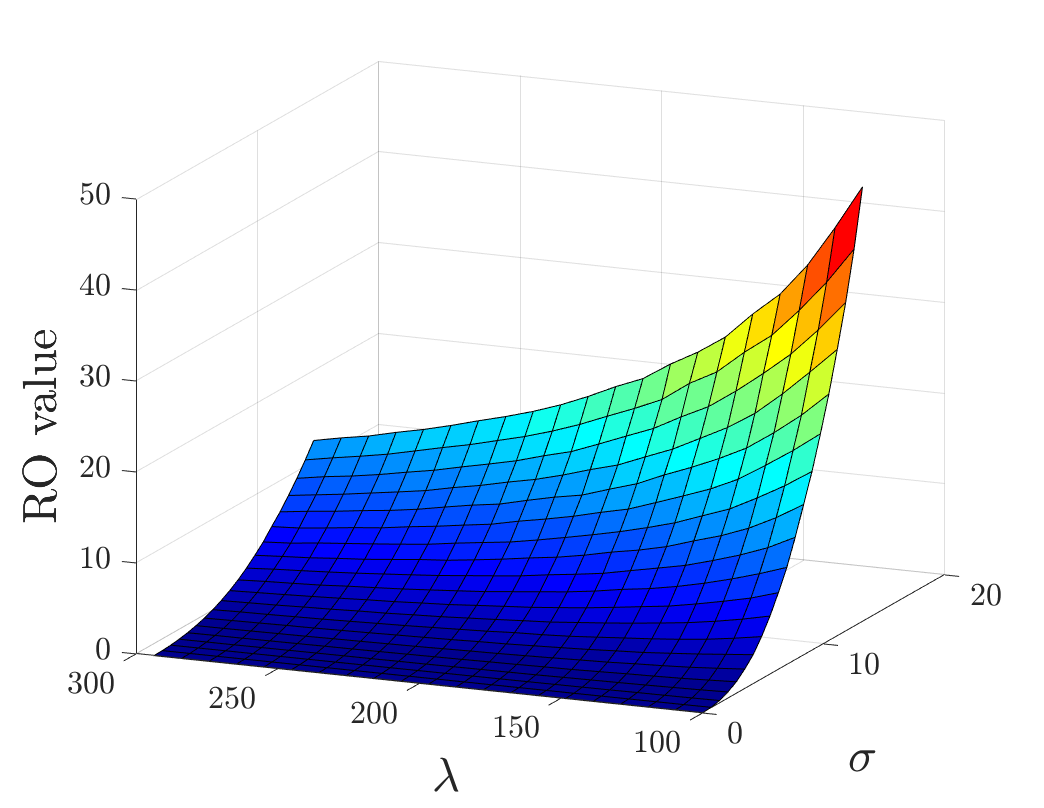}	
	\end{subfigure}
		\caption{Sensitivity analysis of the results using a yearly \(\sigma\) in the range \( (0;2\hat{\sigma}]\) with a strike price $K$ in the range $[20;60]$ (left panel), and a yearly \(\sigma\) in the range \( (0;2\hat{\sigma}]\) with and a yearly \(\lambda\) in the range \( (100;2\hat{\lambda}]\) (right panel). The RO value is expressed in \EUR/MWh.\label{fig:OU_sens}}
	
\end{figure}

Figure \ref{fig:OU_sens} shows the comparative statics for different ranges for the parameters \(\sigma\) and \(\lambda\) and strike price $K$. 
As expected, the higher the strike price, the lower the value of the reliability option for each value of $\sigma$ (left panel). 
On the other hand, both the left and right panels show that, when $\sigma$ increases, the RO value rises as well. Moreover, when $\lambda$ is low, the relative increase in the RO value is high (right panel). This is consistent with the fact that a low $\lambda$ allows fluctuations of the underlying that are far from the long term mean to be more persistent.  

\subsection{Electricity spot price and RO strike price as correlated OU with seasonality}

We simulate now the value of the RO using the model described in Section \ref{Sec:season_2}, again by means of a Monte Carlo method (again using 10,000 runs). We start from a given correlation coefficient, set at \(\rho=0.5\), and assume that \(\lambda_K\) and \(\sigma_K\) are equal to the ones estimated for the electricity price and $X_0 = 0$. 
In line with the PUN mean price, which is equal to 42.77 \EUR/MWh, $K_0$ is arbitrarily chosen equal to 40 \EUR/MWh, so that, after de-seasonalizing (using the same estimated seasonality parameters of the PUN price), we obtain $Y_0 = -0.21$. 

\begin{figure}[h!]
	\centering
	\begin{subfigure}[b]{0.45\textwidth}
	\includegraphics[width=\textwidth]{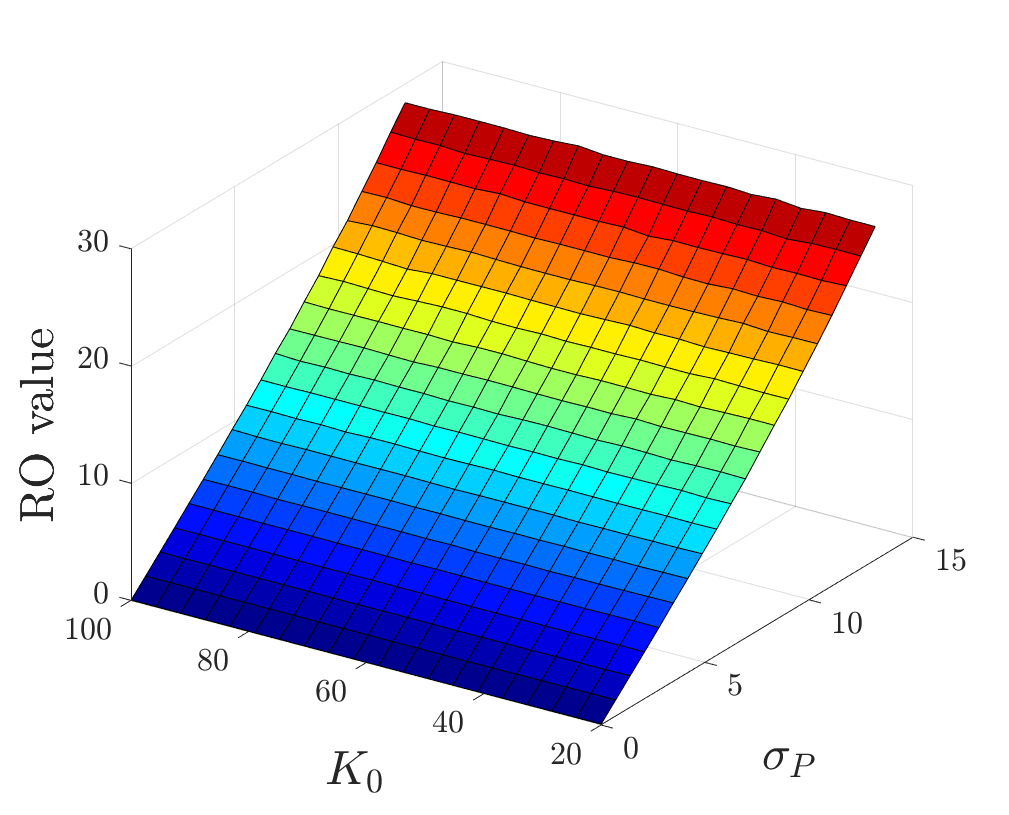}
	\end{subfigure}
	\quad
	\begin{subfigure}[b]{0.45\textwidth}
	\includegraphics[width=\textwidth]{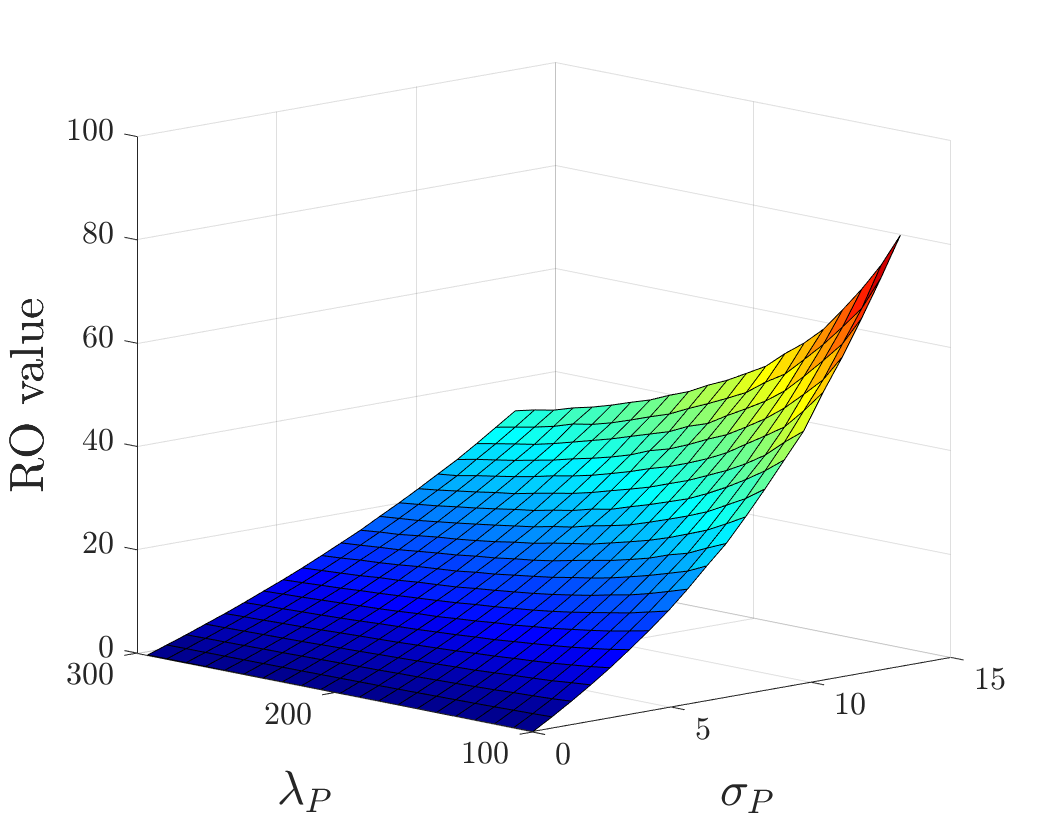}
	\end{subfigure}
	\quad
	\begin{subfigure}[b]{0.45\textwidth}
	\includegraphics[width=\textwidth]{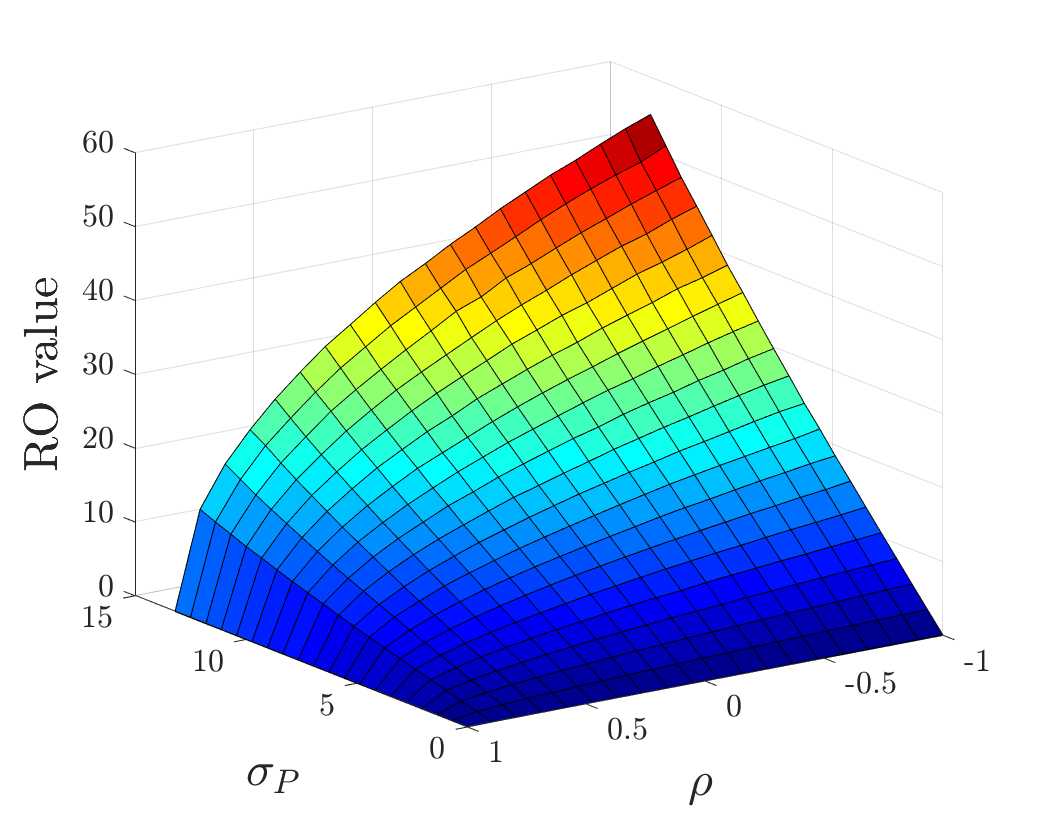}
	\end{subfigure}
	\quad
	\begin{subfigure}[b]{0.45\textwidth}
	\includegraphics[width=\textwidth]{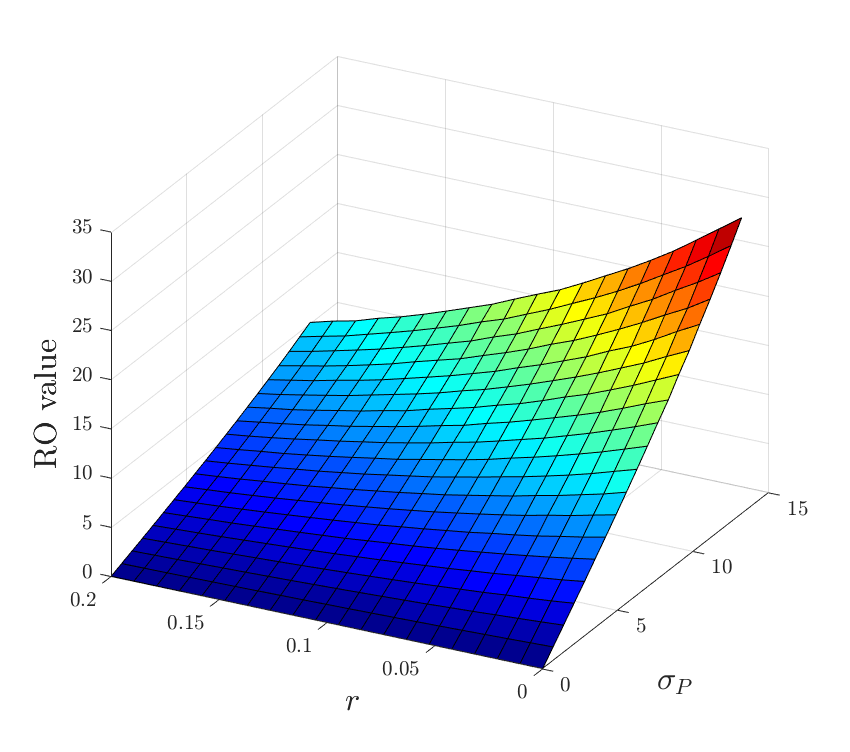}
	\end{subfigure}
	\caption{Sensitivity analysis of the results using a yearly \(\sigma_P\) in the range \( (0;2\hat{\sigma}_P]\) with an initial strike $K_0$ in the range \( [20;100]\) (upper left panel), with a yearly \(\lambda_P\) in the range \( (100;2\hat{\lambda}_P]\)(upper right panel),  with a correlation \(\rho\) in the range \( [-1;1]\) (left bottom panel) and with a yearly risk free rate \(r\) in the range \( [0;0.2]\) (right bottom panel).}
	\label{fig:2OU_sens}
\end{figure}

Figure \ref{fig:2OU_sens} shows the results when we assume the strike price process to have the same parameters estimated for the electricity price $P$. 
The upper left panel shows that the initial level of the strike price $K_0$ has no influence on the value of the reliability option. This is due to the magnitude of the estimated $\lambda_P$, and thus of $\lambda_K$: a mean reversion speed as high as that estimated makes the strike price process return to its mean level in an amount of time negligible with respect to the maturity. This implies that the starting point of the process has no relevant impact on the RO value.

The upper right panel of Figure \ref{fig:2OU_sens} instead shows how sensitive the RO value is to changes in the electricity price parameters $\lambda_P$ and $\sigma_P$ (and thus in turn in $\lambda_K$ and $\sigma_K$). Similarly to what we have observed  before, the higher the volatility of the underlying (and, in this case, of the strike price), the higher the RO value. This relationship increases in  proportionality as the speed of mean reversion decreases, since it takes more time to return to the mean, and thus volatility matters more.  

The impact of the correlation factor $\rho$ is instead investigated in the bottom left panel, where we assess how different correlation factors in the range $[-1;1]$ affect the price of the reliability option. When the two assets are perfectly correlated ($\rho =1$), the RO value is zero for all levels of $\sigma_P$. In fact, as seen in Section \ref{Sec:GBM_corr}, the volatility is minimized and the RO can be interpreted as an integral of calls, with maturity ranging in the interval $[T_1,T_2]$, being exactly at the money at the time of expiration, and thus having zero value. Instead, as shown, when the two processes are uncorrelated, the level of risk increases, and it reaches its maximum when they are perfectly negatively correlated. In this case, the volatilities of the two Brownian motions sum up, increasing the volatility of the option payoff and minimizing the risk of having the calls at the money. 
Finally, the bottom right panel shows that the RO price is negatively correlated with the risk free rate $r$: a higher $r$ decreases the option value as it lowers the discounted cash flows. 

In the previous figures, the parameters for \(\lambda_P\) and \(\lambda_K\), and \(\sigma_P\) and \(\sigma_K\), were tied together, in the sense that \(\lambda_K\) and \(\sigma_K\) were always equal to, respectively, \(\lambda_P\) and \(\sigma_P\).
Instead, we now investigate what happens when \(\sigma_K\) equals \(\sigma_P\) as before, but \(\lambda_K\) changes independently from \(\lambda_P\). Moreover, we also investigate the effects of a variation in $\sigma_P$ different from that in $\sigma_K$. Figure \ref{fig:2OU_sens_single_params} and \ref{fig:2OU_sens_sigmas_rho} show the results. 

The left panel of Figure \ref{fig:2OU_sens_single_params} reports the results for a variation in $\lambda_K$ (in the range \( (0;2\hat{\lambda}_P]\) and  shown in $\log_{10}$ scale) independent from the value of $\lambda_P$. The graph shows how $K_0$ hardly affects the $RO$ value, as it has an impact only when both \(\sigma_K\) and $\lambda_K$ are sufficiently small. This confirms the result shown above that the initial condition of the parameters matters only when it takes a sufficient amount of time for them (i.e., for the strike price in this case) to return to their long term value.
The right panel instead shows the sensitivity of the RO value to changes in the yearly \(\lambda_K\) (again in the range \( (0;\hat{\lambda}_P]\)) independent from the value of $\lambda_P$, and in the correlation factor \(\rho\) (in the range \( [-1;1]\)) (in this graph, \(\sigma_K\) is always equal to \(\sigma_P\) and they are in turn equal to \(\hat{\sigma}_P\), \(\lambda_P=\hat{\lambda}_P\), and \(\lambda_K\) is shown in $\log_{10}$ scale.). Here, the $\rho$ value matters more when both $\lambda_K = \lambda_P$ and \(\sigma_K = \sigma_P\). In fact, $\rho$ (negatively) affects the RO value only when it tends to $-1$ and \(\lambda_K\) is closer to the value of \(\lambda_P\) (note that, in the figure, \(\lambda_K \in (0;\hat{\lambda}_P]\), where \( \hat{\lambda}_P \) corresponds to the value of \(2.47\) in \(\log_{10}\) scale). This confirms our intuition that, when the initial value of the electricity price and the strike price are close and the two random variables follow the same dynamics, the RO has a negligible value since it is likely that it will be always at-the-money. Conversely, if the two random variables are not perfectly correlated or the two variables follow different dynamics, it is unlikely that at every point in time $P_t$ and $K_t$ coincide, and this adds value to the RO.  

\begin{figure}[H]
	\centering
	\begin{subfigure}[b]{0.45\textwidth}
	\includegraphics[width=\textwidth]{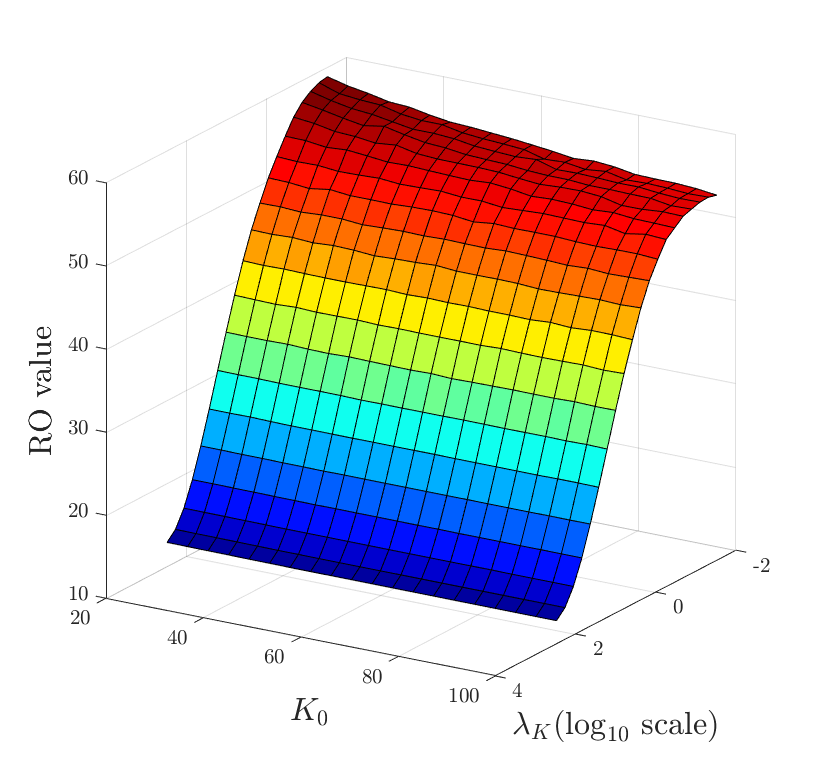}
		\end{subfigure}
		\quad
	\begin{subfigure}[b]{0.45\textwidth}
	\includegraphics[width=\textwidth]{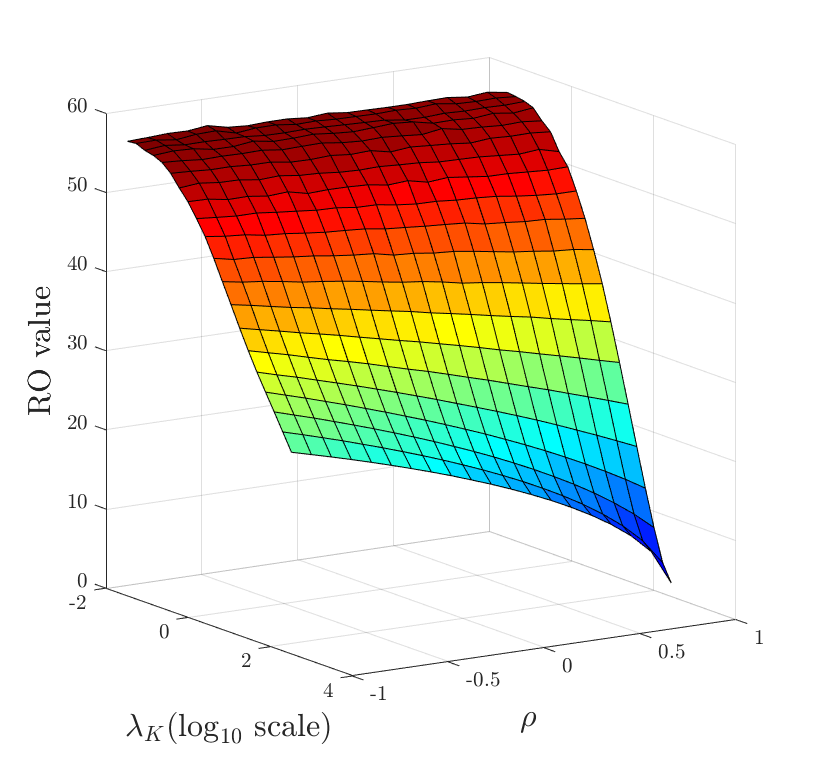}
		\end{subfigure} 
		\caption{Sensitivity analysis of the results using a yearly \(\lambda_K\) in the range \( (0;\hat{\lambda}_P]\) with an initial strike price $K_0$ in the range $[20;100]$, both with a yearly \(\sigma_K\) equal to the yearly \(\sigma_P\) (upper left panel) and with and a scaled down yearly \(\sigma_K\) (upper right panel), and with a correlation \(\rho\) in the range \( [-1;1]\) (bottom panel) (here \(\sigma_K = \sigma_P\)).  
		The RO value is expressed in \EUR/MWh.}
	\label{fig:2OU_sens_single_params}
\end{figure}

\begin{figure}[h!]
	\centering
	\begin{subfigure}[b]{0.45\textwidth}
	\includegraphics[width=\textwidth]{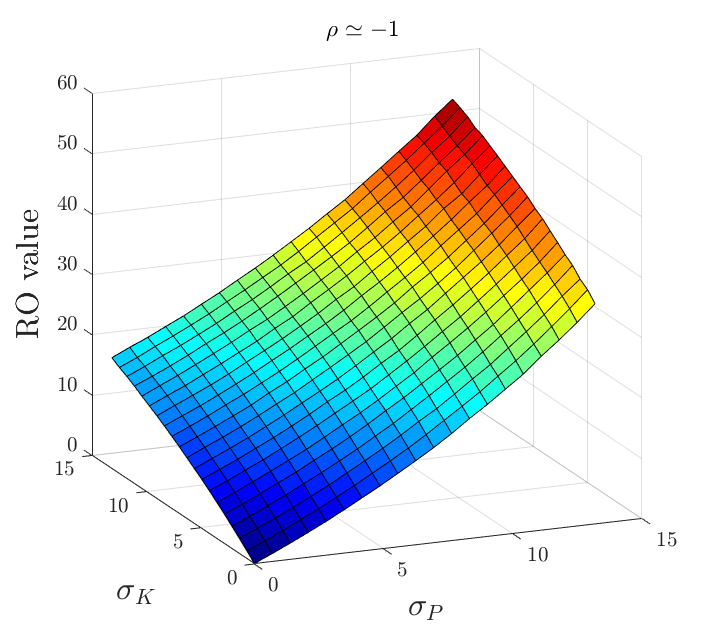}
		\end{subfigure}
	\quad
	\begin{subfigure}[b]{0.45\textwidth}
	\includegraphics[width=\textwidth]{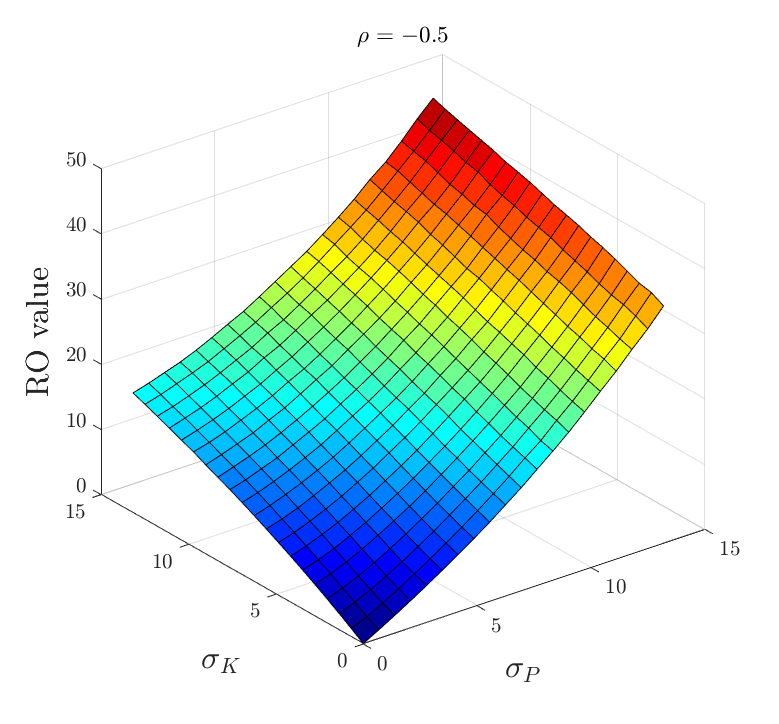}
		\end{subfigure} 
		\quad
	\begin{subfigure}[b]{0.45\textwidth}
	\includegraphics[width=\textwidth]{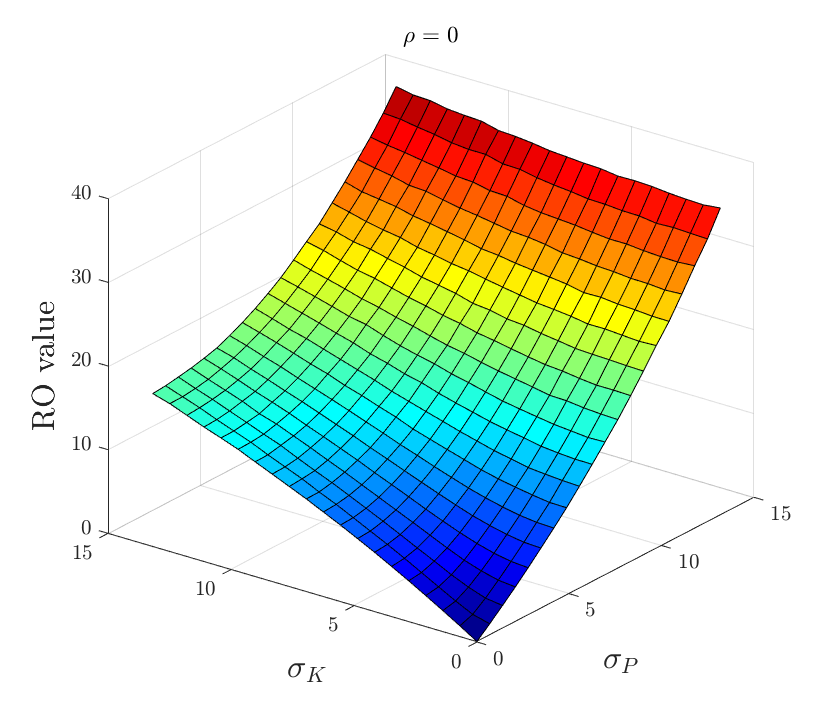}
		\end{subfigure} 
		\quad
	\begin{subfigure}[b]{0.45\textwidth}
	\includegraphics[width=\textwidth]{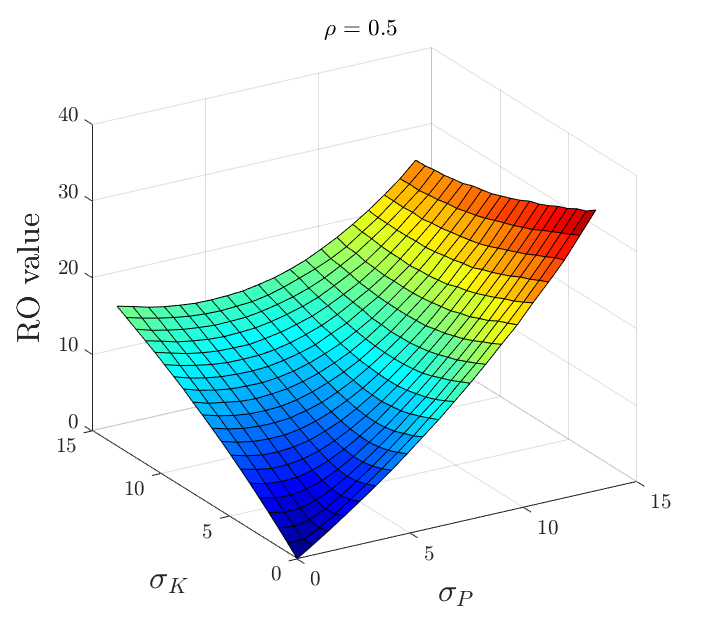}
		\end{subfigure}
	\quad
	\begin{subfigure}[b]{0.45\textwidth}
	\includegraphics[width=\textwidth]{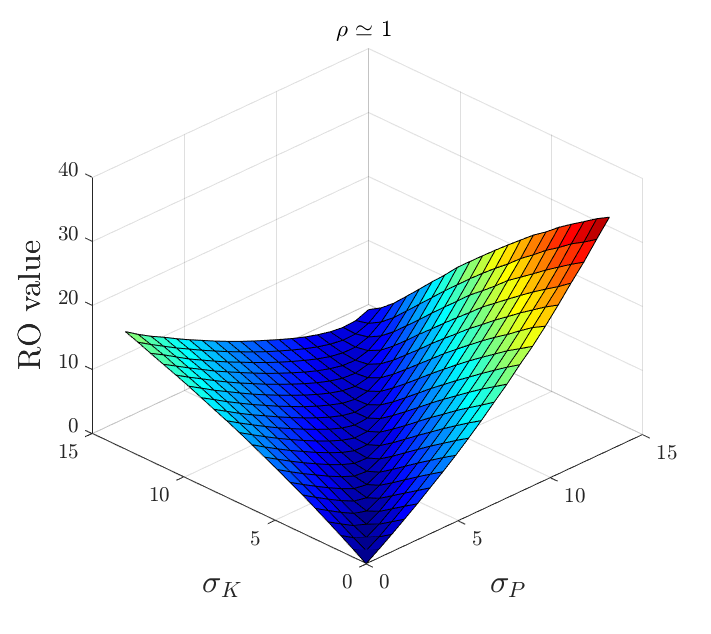}
		\end{subfigure}
		\caption{Sensitivity analysis of the RO value to a disjoint variation in the two volatilities, with a yearly \(\sigma_P\) and \(\sigma_K\) in the range \( (0;2\hat{\sigma}_P]\) (here \(\lambda_K = \lambda_P\)). In the different panels, we can see how a variation in the correlation coefficient $\rho$ affects the RO value: when the two processes are independent or negatively correlated, higher $\sigma_P$ and $\sigma_K$ result in a higher option value. However, when the correlation is positive (middle right and bottom panels), the higher the correlation, and the more the two volatilities are similar, the lower the value of the option. The RO value is expressed in \EUR/MWh.}
	\label{fig:2OU_sens_sigmas_rho}
\end{figure}

Finally, Figure \ref{fig:2OU_sens_sigmas_rho} shows the effect of a disjoint variation in the two volatilities, with a yearly \(\sigma_P\) and \(\sigma_K\) in the range \( (0;2\hat{\sigma}_P]\), for different levels of $\rho$ (in these graphs, \(\lambda_K\) is always equal to \(\lambda_P= \hat{\lambda}_P\)). When \(\rho \leq 0\), the RO price is always increasing in the electricity price volatility \(\sigma_P\) and in the strike price's one \(\sigma_K\). This is as expected, since volatility adds value to the call options. Instead, when \(\rho > 0\), the fact that the two processes move together can lower the aggregate risk, since the spread between the electricity price and the strike price reduces. This translates into a negative effect on the option value. The RO value is minimized when \(\sigma_P\) = \(\sigma_K\). In Figure \ref{fig:2OU_sens_sigmas_rho}, panel \(\rho = 0.5\), we can see that the option value is still positive; in the panel \(\rho = 1\), the RO value becomes null for $\sigma_P = \sigma_K$, since, as mentioned, if the two processes are perfectly positively correlated, the RO value coincides with its intrinsic value. Thus, there is a non-monotone effect of the volatility increase of one process, depending on the amount of volatility of the other process, and on the level of the correlation coefficient. The inflection is maximum when the two processes are perfectly positively correlated.

\section{Conclusions}\label{sec:conclusion}

In this paper, we have studied the value of the RO from a financial perspective.  The financial approach to option pricing relies on the assumption that a risk-neutral measure exists, which is equivalent to assume that markets are complete. This is not a problem for pricing options on electricity prices, as long as they can be written on electricity futures that can be rolled over the delivery period of the RO. Nevertheless, such an approach does require that RO markets are competitive and that forward markets are liquid. Our analysis provides a benchmark value for the RO, under the assumption that the market for the derivative is liquid enough to bring about competition.\footnote{Note that according to \cite{Bid, CramtonOckenfelsStoft} ROs are instruments that enhance competition in the electricity market.} 
Therefore, the simplified mathematical model that we proposed can be seen as a starting point in the analysis of ROs. We obtain semi-explicit formulae for the value of the RO, under a set of different assumptions with increasing realism and complexity. We move from simple integrals of call options written on GBMs to correlated mean reverting processes that capture the behavior of realistic electricity price time series, on the one hand, and complex rules for RO, on the other. Then, we  simulate the value of the Reliability Option through a real-market calibration of the parameters. 

The results evidence that the value of the RO moves consistently with expectations from option theory: a rise in the strike price lowers the RO value, which depends positively on the volatility of the electricity price, as well as on the volatility of the strike price itself. The mean reversion speed of the processes reduces the impact of the starting point, which was another expected result. However, when both the strike price and the electricity price are assumed to be stochastic processes, the value of the RO depends crucially on their correlation coefficient $\rho$. In particular, a positive correlation reduces the value of the RO. Moreover, there is a non-monotone impact of the volatility of one process, depending on the level of volatility of the other process and on a positive correlation. This is important when designing the rule of the RO. For instance, if the strike price is allowed to change with respect to a reference marginal cost, which is also believed to be the technology setting the system marginal price at the day ahead level, the two process clearly covariate positively. In this case, it is very likely that a RO has a very limited value, for every possible starting value of the state variables $P$ and $K$. More in general, our results show that a careful estimate of the parameters is needed to calculate the value of the ROs. \emph{Ceteris paribus}, the RO value will be lower as the volatility of the electricity price decreases, the strike price increases, the speed of mean reversion increases, the correlation of the electricity price with the strike price increases (if the strike price is allowed to change over time), and the two volatilities are closer. These are all factors that need to be taken into account when designing the market for ROs and calculating the equilibrium value. 
\bigskip

\noindent
\textbf{Acknowledgments}. 
Part of this work was done while the first author had a post-doc position under the financial grant ``Capacity markets and the evaluation of reliability options" from the Interdepartmental Centre for Energy Economics and Technology ``Giorgio Levi-Cases", University of Padua, which she gratefully acknowledges, and while the second author was completing her PhD at the University of Padua. The authors also wish to thank Paolo Falbo, Giorgio Ferrari, Michele Moretto, Filippo Petroni, Marco Piccirilli and Dimitrios Zormpas for interesting discussions, and all the participants to the 2nd Conference on the Mathematics of Energy Markets in Vienna, XLI A.M.A.S.E.S. Conference in Cagliari, Energy Finance Christmas Workshop 2017 in Cracow, XIX Workshop in Quantitative Finance in Rome, Energy Finance Italia 3 in Pescara, 3rd C.E.M.A. Conference in Rome, XLII A.M.A.S.E.S. Conference in Naples.


\appendix
\section*{Appendix} 
\setcounter{theorem}{0}
\renewcommand{\thetheorem}{A.\arabic{theorem}}
\setcounter{equation}{0}
\renewcommand{\theequation}{A.\arabic{equation}}
\setcounter{figure}{0}
\setcounter{table}{0}
\renewcommand{\thetable}{A.\arabic{table}}
\renewcommand{\thesubsection}{A.\arabic{subsection}}

\subsection{Proofs of pricing formulae}

\begin{proof}[Proof of Proposition~\ref{Prop:GBM}]
 The quantity $f(s,\omega)\colon =e^{-rs}Q (P_s(\omega) - K)^+$ in Equation \eqref{form_of_price} is non-negative. 

Then, if we set
\begin{equation}\label{expect}
A(K,P_{0},s)\colon =  e^{-r s}\mathbf{E}^{\mathbb{Q}}\left[\left.\left(P_s-K\right)^+\right|\mathcal{F}_{0}\right]\,,
\end{equation}
by Tonelli's theorem, we get
\begin{equation}\label{Call_option_after_fubini}
RO(T_1,T_2)=Q\int_{T_1}^{T_2}A(K,P_{0},s)ds\,.
\end{equation}
$A(K,P_{0},s)$ is clearly the price of a European call option with strike price $K$ and maturity $s$, thus Equation~\eqref{GBM_price_formula_BS} is simply obtained  with the Black~and~Scholes formula. 
\end{proof}

\begin{proof}[Proof of Proposition~\ref{Prop:corr_GBM}]
As in the proof of Proposition~\ref{Prop:GBM}, if we write
\begin{equation}\label{expect}
A(K_{0},P_{0},s)\colon =  e^{-rs}\mathbf{E}^{\mathbb{Q}}\left[\left.\left(P_s-K_s\right)^+\right|\mathcal{F}_{0}\right]\,,
\end{equation}
then, by Tonelli's theorem, we have
$$
RO(T_1,T_2) = Q \int_{T_1}^{T_2} A(K_{0},P_{0},s)ds\,.$$
Here, $A(K,P_{0},s)$ is the price of an exchange option between the electricity price $P$ and the strike price $K$, with maturity $s$, thus Equation~\eqref{price_corr_GBM} is simply obtained  with the Margrabe formula with dividends (see \cite{CarDur}).

\end{proof}

\begin{proof}[Proof of Proposition~\ref{Prop:OU}] 
As in the previous proofs, we write $A(K,P_{0},s)\colon =  e^{-r s}\mathbf{E}^{\mathbb{Q}}\left[\left.\left(P_s-K\right)^+\right|\mathcal{F}_{0}\right]$ and apply Tonelli's theorem to obtain
$$ RO(T_1,T_2)=Q\int_{T_1}^{T_2}A(K,P_{0},s)ds\,.$$
We now notice that
$$ A(K,P_0,s) = e^{-rs} \mathbf{E}^{\mathbb{Q}}\left[ \left.\left(f(s,s) - K\right)^+ \right|\mathcal{F}_{0}\right] $$
where $f(t,s)$, $t \in [0,s]$, has the dynamics
$$ df(t,s) = f(t,s) \sigma e^{-\lambda(s-t)}\ dW_t $$
The result then follows from the Black-Scholes formula with time-dependent (deterministic) volatility, which enters into the formula via the integral of its square, here equal to
$$ \int_0^s  \left( \sigma e^{-\lambda(s-t)} \right)^2 \ dt = \frac{\sigma^2}{2\lambda} (1 - e^{-2\lambda s}) = Var(s) $$ 

Equation~\eqref{OU_price} follows.
\end{proof}

\begin{proof}[Proof of Proposition~\ref{Prop:2OU}]
As before, we write $A(P_{0},K_0,s) \colon = e^{- r s}  \mathbf{E}^{\mathbb{Q}}\left[\left. ({\color{red} P_s} - {\color{red} K_s})^+\  \right| \mathcal{F}_{0} \right]$, we use Tonelli's theorem and obtain
$$ RO(T_1,T_2)=Q\int_{T_1}^{T_2}A(K,P_{0},K_0,s)ds \,.$$
 Now, as in the proof of Proposition 3.3, we now notice that
$$ A(K,P_0,s) = e^{-rs} \mathbf{E}^{\mathbb{Q}}\left[ \left.\left(f_P(s,s) - f_K(s,s)\right)^+ \right|\mathcal{F}_{0}\right] $$
where  

$f_i(t,s)$, $t \in [0,s]$, $I = P,K$, have the dynamics
\begin{eqnarray*}
df_P(t,s) & = & f_P(t,s) \sigma_x e^{-\lambda_x(s-t)}\ dW^1_t, \\
df_K(t,s) & = & f_K(t,s) \sigma_y e^{-\lambda_y(s-t)}\ dW^2_t, 
\end{eqnarray*}
The result then follows from the Margrabe formula with time-dependent (deterministic) volatilities, which now enters into the formula via the integral of the squared volatility of $f_p(\cdot,s)/f_K(\cdot,s)$ (see e.g. \cite{Deng}), here equal to
$$ \int_0^s \left( \sigma^2_x e^{-2\lambda_x(s-t)} + \sigma^2_y e^{-2\lambda_y(s-t)} - 2 \rho \sigma_x \sigma_y e^{-(\lambda_x + \lambda_y) (s-t)}\right) \ dt = \overline{Var}(s) $$

Equation \eqref{2OU_price} follows.

\end{proof}

\end{document}